\documentclass[lettersize,journal]{IEEEtran}
\usepackage{amsmath,amsfonts}
\usepackage{algorithmic}
\usepackage{algorithm}
\usepackage{array}
\usepackage[caption=false,font=normalsize,labelfont=sf,textfont=sf]{subfig}
\usepackage{textcomp}
\usepackage{stfloats}
\usepackage{url}
\usepackage{verbatim,cancel}
\usepackage{graphicx}

\usepackage[url,hyperrefblack,notheorems,IEEEtran]{research17} 
\usepackage{soul}
\usepackage{dutchcal}
\usepackage{xcolor}
\usepackage{amsthm}
\usepackage{nicematrix}
\usepackage{arydshln}
\usepackage{tikz}
\usetikzlibrary{fit}

\newtheorem{theorem}{\mytheoremname}
\newtheorem{lemma}{\mylemmaname}

\newtheorem{conjecture}{\myconjecturename}
\newtheorem{proposition}{\mypropositionname}

\newtheorem{definition}{\mydefinitionname}
\newtheorem{remark}{\myremarkname}
\newtheorem{example}{\myexamplename}

\newcommand*{\Scale}[2][4]{\scalebox{#1}{\ensuremath{#2}}} 

\newcommand\undermat[2]{%
  \makebox[0pt][l]{$\smash{\underbrace{\phantom{%
    \begin{matrix}#2\end{matrix}}}_{\text{$#1$}}}$}#2}
\NiceMatrixOptions{code-for-first-row = \scriptstyle,code-for-first-col = \scriptstyle }

\definecolor{macaroniandcheese}{rgb}{1.0, 0.74, 0.53}
\definecolor{babyblue}{rgb}{0.54, 0.81, 0.94}
\definecolor{celadon}{rgb}{0.67, 0.88, 0.69}

\newcommand{\bs}[1]{\ensuremath{\boldsymbol{#1}}}

\newcommand{\vol}[1]{\operatorname{vol}\left(#1\right)} 
\newcommand*{\medcap}{\mathbin{\scalebox{1.5}{\ensuremath{\cap}}}}%
\newcommand*{\medcup}{\mathbin{\scalebox{1.5}{\ensuremath{\cup}}}}%

\definecolor{macaroniandcheese}{rgb}{1.0, 0.74, 0.53}
\definecolor{babyblue}{rgb}{0.54, 0.81, 0.94}
\definecolor{celadon}{rgb}{0.67, 0.88, 0.69}

\begin{document}

\title{Coset Shaping: Constructions and Bounds}

\author{Irina Bocharova~\IEEEmembership{Senior Member,~IEEE}, Maiara F. Bollauf~\IEEEmembership{Member,~IEEE}, \\ and Boris Kudryashov~\IEEEmembership{Senior Member,~IEEE} 
\thanks{This work was supported by the Estonian Research Council, grant PRG2531.}
\thanks{All authors are affiliated to the University of Tartu, Narva mnt 18, Tartu 51009, Estonia.
        Their respective e-mail addresses are irina.bocharova@ut.ee, maiara.bollauf@ut.ee, and boris.kudryashov@ut.ee.}
\thanks{This paper is to be published partially in the proceedings of the International Zurich Seminar on Information and Communication (IZS), Zurich, Switzerland, February 2026~\cite{bocharova2025coset}.}}



\maketitle

\begin{abstract}
A new geometric shaping technique, referred to as \emph{coset shaping}, is proposed and analyzed for coded QAM and PAM signaling. This method can be applied to both information and parity bits without introducing additional complexity. It is shown that, as the error-correcting code length and the modulation order grow, the gap to capacity of the proposed shaping scheme can be made arbitrarily small. A Gallager-type bound is provided together with numerical results, including  performance comparisons for the shaping scheme combined with short and mid-length binary-coded, as well as nonbinary LDPC-coded QAM signaling.


\end{abstract}


\section{Introduction}
\label{sec:introduction}

\IEEEPARstart{A}{n} asymptotic loss relative to the Shannon signal-to-noise (SNR) limit of reliable coded modulation transmission with uniformly distributed input signals over the AWGN channel reaches 1.53 dB ({\em ultimate limit}). If both code length and signal constellation size tend to infinity, this gap can be reduced by using a Gaussian distribution on the input signals. Shaping in an AWGN channel refers to transmitting low-energy signals more frequently than high-energy to mimic Gaussian signaling.

It is generally assumed that for QAM modulation with a finite alphabet size, the optimal input distribution is an approximation of the Maxwell-Boltzmann distribution.  A few shaping methods are used in practice. These techniques directly or indirectly implement optimal or near-optimal nonuniform input signaling.

The so-called {\em geometric} shaping uses equiprobable, non-uniformly spaced signal points, with a preference for low-energy points. It was studied, for example, in \cite{sun1993approaching,quzhen2017}. 
Multidimensional geometric shaping can be efficiently implemented in the form of {\em Voronoi shaping} (see, for example, \cite{buglia2021voronoi, steiner2017comparison,zhou2017shaping,sheng2024multidimensional,li2025coded}). In particular, in  \cite{sheng2024multidimensional,li2025coded},  Voronoi shaping for multilevel coded modulation was analyzed. 

{\em Probabilistic amplitude shaping} (PAS) was studied, for example, in \cite{bocherer2015bandwidth, schulte2015constant, gultekin2018constellation, fehenberger2015ldpc,
steiner2017ultra}. It uses uniformly spaced signal points with probabilities determined by the nonuniform distribution. Implementations of PAS in the form of enumerative decoding of constant-composition codes \cite{bocherer2015bandwidth} and enumerative sphere shaping codes \cite{gultekin2018constellation,gultekin2020probabilistic} were considered in combination with binary LDPC error-correcting codes. In these schemes,  a \emph{distribution matcher} (DM)
maps information bits to shaped bits, which are then systematically encoded by appending uniformly distributed parity bits. Linear layered probabilistic shaping, which extends PAS to probabilistic shaping of parity bits, was introduced in \cite{bocherer2019probabilistic} and further developed in \cite{lentner2025efficient}. The shaping scheme in \cite{bocherer2019probabilistic} processes codewords of a linear block code with
two DMs. The message bits are shaped with a conventional 
DM and encoded systematically, and a so-called syndrome DM shapes the parity bits. Nonbinary (NB) LDPC codes of moderate lengths over GF$(2^6)$ and GF$(2^8)$ used in conjunction with PAS were analyzed in \cite{steiner2017ultra}. 
PAS applied after encoding with moderate- and long-length (NB) LDPC codes was considered in \cite{bocharova2024analysis} and \cite{bocharova2023ldpc}, respectively.

In \cite{bocharova2025coset}, we proposed 
a new competitive shaping technique, called {\em coset shaping}.
The new scheme uses a combination of a linear error-correcting code (ECC) with a
higher rate than the target transmission rate and a shaping
code. The shaping code is a set of coset leaders of a linear
code. The best coset leader is selected using the minimum signal energy criterion. Similar to \cite{bocherer2019probabilistic}, the proposed shaper allows one to shape both information and parity bits. For the proposed scheme, the gap-to-capacity can be made arbitrarily small as the length of the ECC and the modulation order approach infinity. In \cite{bocharova2025coset}, the sketch of the corresponding proof is presented.

In this paper, we expand on the asymptotic analysis of the coset shaper. A modified Gallager-type upper bound on the probability of maximum-likelihood (ML) decoding error for finite-length coded, shaped QAM signaling is presented and compared with the Gallager random coding bound (RCB).   
The presented bounds shed light on the choice of signal probability distribution for shaping. The bounds computed for different code lengths show that the capacity-achieving distribution is far from optimal in some scenarios.

This paper is organized as follows. Preliminaries are in Section~\ref{prelim}, and the coset shaping scheme is described in Section \ref{sec:coset}. The analysis of the convergency of the coding gain to its ultimate limit is given in Section \ref{analysis}. The Gallager bound and its modification are presented and discussed in Section  \ref{bounds}. Simulation results and comparisons are given in Section \ref{sec:numres}. 

\section{Preliminaries \label{prelim}}
\label{sec:preliminaries}

\subsection{Notation}
\label{sec:notation}

We denote by $\Naturals$, $\Integers$, and $\Reals$ the set of naturals, integers, and reals, respectively. An interval over $\Integers$ is $[i:j]\eqdef\{i,i+1,\ldots,j\}$ for $i,j\in \Integers$, $i\leq j$. Vectors are \emph{row} vectors and boldfaced, e.g., $\vect{x}$. Matrices and sets are given by capital sans serif letters and calligraphic uppercase letters, respectively, $\mat{X}$ and $\set{X}$. An identity matrix $n \times n$ is denoted by $\mat{I}_n$. The addition in $\Field_2$ is denoted by $\oplus$, and $+$ is the real addition.

Any $M^2$-QAM signal sequence of length $n$ can be interpreted as an $M$-PAM signal sequence of length $2n$. In this paper, we use both interpretations interchangeably. In the description of the shaping method and asymptotic analysis, we mainly use PAM signaling. In simulations and finite-length comparisons with other techniques, we consider QAM signaling.

\subsection{Modulation and Mapping for Shaped Coded PAM}
\label{sec:modulation-mapping}

Consider $2^m$-PAM modulation. Let $n_{\rm s}$ denote the codelength in signals.
Each of  $2^m$ signal points is indexed by an $m$-bit binary sequence. 
A PAM modulator is used to convert a codeword of length $n=mn_{\rm s}$ bits into $n_{\rm s} = \nicefrac{n}{m}$ PAM signals
\[x_t=(2s_t-1)A(a_{1t},\dots,a_{(m-1)t}),\] where $t = 1,2,\dots,n_{\rm s}$, $s_t$ is sign bit, $a_{1t},\dots, a_{(m-1)t)}$ are amplitude bits, and
$A(\cdot)$ denotes the amplitude of the signal point. Matching of code symbols with unshaped and shaped PAM signals can significantly improve the decoding performance of coded modulation systems. For instance, this holds for nonbinary (NB) LDPC-coded unshaped and shaped PAM signals used with the belief propagation (BP) decoding ~\cite{declercq2004regular, bocharova2023nonbinary, bocharova2024analysis}. 

We consider a mapping of binary codewords to the sequences of $2^m$-PAM signals which facilitates the {\em coset shaping} proposed in this paper. 

Let $\code{C}$ be  a linear $[n,k]$ binary code
with a generator matrix
\begin{equation*}
\label{c1}§
\mat{G}=
\begin{pNiceMatrix}
\mat{G}_{m-1} & \mat{G}_{m-2}& ... & \mat{G}_{0} 
\end{pNiceMatrix},    
\end{equation*} 
where the submatrices $\mat{G}_i$ are of size $k\times n_{\rm s}$, $n_{\rm s}=\nicefrac{n}{m}$, for $i \in [0:m-1]$, where $m$ is a divisor of $n$.

It is assumed that each of $\mat{G}_i$ 
determines the 
$i$-th bit in the $m$-bit long representation of the PAM signal $x_t$. In other words, $\mat{G}_0$ determines sign bits $s_t$ of $n_{\rm s}$ PAM signals, and $\mat{G}_i$ determine the corresponding amplitude bits, for $i \in [1:m-1]$.  

For a given message $\bs u \in \mathbb F_2^k$, the codeword $\bs v \in \code{C}$ is 
\begin{equation}
\label{c2}
\bs v =
\begin{pNiceMatrix}
\bs v_{m-1} & \bs v_{m-2} & ...& \bs v_{0}
\end{pNiceMatrix},
\quad \bs v_i=\bs u \mat{G}_i, 
\end{equation} 
$\bs v_i \in \mathbb F_2^{n_{\rm s}}, i \in [0:m-1]$. 
Then $\bs v$ is rewritten in the form of a $m\times n_{\rm s}$ matrix $\trans{\begin{pNiceMatrix} \bs v_0 ~ \bs v_1 ~ \dots  \bs v_{m-1} \end{pNiceMatrix}}$ whose $i$-th column is used as the Gray-coded binary index of $2^m$-PAM signal point $x_i \in \mathcal X=\{-2^m+1, \dots, -3, -1, 1, 3, \dots, 2^m-1\}$.

\begin{example} 
\label{ex1}
Let $m=3$, $n=6$, $n_{\rm s}=2$. A generator matrix for a $[6,3]$-code $\code{C}$ in systematic form is
\begin{equation}\label{shorient}
\mat G= \begin{pNiceMatrix}
    1 &0 &  0 &1 & 1 & 0 \\
    0 & 1 &  0 &1 & 0 & 1 \\
    \undermat{{\mat G}_2}{0 & 0} & \undermat{{\mat G}_1}{1 & 0} & \undermat{{\mat G}_0}{1 & 1}
\end{pNiceMatrix}.
\end{equation}
\vspace{0.3cm}

Let $\bs u=(1,1,0)$, the corresponding codeword $\bs v= \bs u \mat{G}  = (1, 1, 0, 0, 1, 1)$  used with $2^3$-PAM. First, we split the codeword into $m=3$ parts and place the equal-length vectors in the rows of a matrix of size $m \times n_{\rm s} = 3 \times 2$, as
\begin{equation*}
\begin{pNiceMatrix}
\bs v_0\\
\bs v_1\\
\bs v_2
\end{pNiceMatrix}
=
\begin{pNiceMatrix}
1&1\\
0&0\\
1&1
\end{pNiceMatrix}.
\end{equation*}
Then columns of this matrix are mapped to the PAM signal sequence $\bs s=(5,5)$ according to bit reflected binary Gray (BRBG) code~\cite{agrell2004optimality} (see Table \ref{tab1}). \hfill \exampleend
\end{example}

\begin{table}[!t]
\label{tab1}
    \caption{BRBG code for 8-PAM}
\begin{tabular}{|c|c|c|c|c|c|c|c|c|}
\hline
 BRBG    &000&001&011&010&110&111&101&100\\
 \hline
 8-PAM  &-7&-5&-3&-1&1&3&5&7 \\\hline 
\end{tabular}
\end{table}

\begin{definition}
The above one-to-one mapping of length $n=mn_{\rm s}$ codewords $\bs v \in \Field_2^n$ to 
signal sequences $\bs s \in \set{X}^{n_{\rm s}}$ is denoted as $\psi$. The sequence $\bs s=\psi (\bs v)$ is called {\em PAM-image of $\bs v$}. A set $\psi(\code{C})$ is called {\em PAM-image of the code}.  
\end{definition}

\section{Coset shaping 
\label{sec:coset}}

We introduce the \emph{coset shaping} scheme for coded {$2^m$-PAM} modulation in its general form next.
The proposed coset shaping scheme is shown in Fig.~\ref{fig:preshaping}, along with the detailed notation in Table~\ref{tab:notations}. Unlike the combinatorial shaping scheme studied in \cite{bocherer2015bandwidth, gultekin2018constellation,gultekin2020probabilistic}, our scheme is based on a linear code and also allows for shaping of the code parity-check bits. At the same time, it provides better performance than linear layered shaping in \cite{bocherer2019probabilistic, lentner2025efficient} since the shaping procedure uses the squared Euclidean norm as the design criterion. The latter leads to an increase in computational complexity, but, as will be shown, sufficiently short shaping codes used in conjunction with ECC of varying lengths, including very long ones,  give a significant gain in decoding performance. The possibility of using this procedure with long ECC implicitly suggests that the complexity of this technique is relatively low compared to the error-correcting code decoding complexity.  

As mentioned above, for PAM-modulation of order $M=2^m$, one of the bits assigned to each signal is a sign bit (s-bit), and the other $m-1$ bits are amplitude bits (a-bits). Signal energy is determined solely by its amplitude. Therefore, a possible energy gain can be obtained by allocating additional redundancy for a desirable (non-uniform) probability distribution on signal amplitudes. Since the parity bits of linear error-correcting codes are always uniformly distributed, we associate information bits mainly with a-bits and parity bits mainly with s-bits. Depending on the target rate, the message length can be larger or smaller than the total number of a-bits. Thus, sometimes a part of the message bits can be transmitted as s-bits, while in some scenarios the parity bits can be assigned with a-bits.         

\label{shap}
\begin{figure}[!t]
\centering
\Scale[0.95]{\includegraphics[width=100mm]{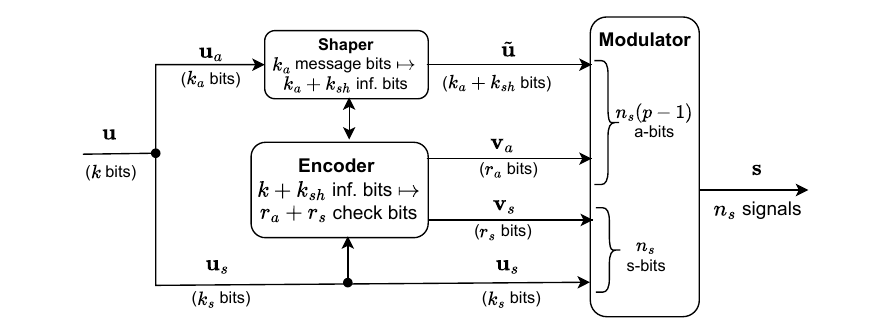}}%
\caption{Shaping scheme.}
\label{fig:preshaping} 
\end{figure}

\begin{table}[!t]
    \centering
\caption{Notations used in Fig.~\ref{fig:preshaping}}
    \label{tab:notations}    
    \begin{tabular}{m{3.2cm}|m{4.6cm}}
    \hline
    Notation & Meaning \\ \hline 
$k_{\rm a}$ ($k_{\rm s}$) & number of message bits transmitted as \\
&amplitude (resp. sign) bits \\ 
$r_{\rm a}$ ($r_{\rm s}$) & number of parity bits transmitted as \\
&amplitude (resp. sign) bits\\    
$k_{\rm sh} = N_{\rm b} k_{\rm sh}^{\rm b}$   & number of auxiliary bits used for shaping \\ 
$k=k_{\rm a}+k_{\rm s}$   & number of message bits \\
$n=mn_{\rm s}$  &   length of \([n,k_{\rm c}]\)-code \\ 
$k_{\rm c}=k+k_{\rm sh}$ & dimension of $[n,k_{\rm c}]$-code \\
$R_{\rm c}=k_{\rm c}/n, \tilde{R}_{\rm c}=mR_{\rm c}$ & rate in bits/code symbol (bits/signal)\\
$R_{\rm T}=k/n, \tilde{R}_{\rm T}=mR_{\rm T}$  & target rate in bits/symbol and signal        \\
$n_{\rm sh}=k_{\rm a}+k_{\rm sh}$ &  length of the $(n_{\rm sh},2^{k_{\rm sh}})$ shaping code \\ 
$R_{\rm sh}=k_{\rm sh}/n_{\rm sh}$&  rate of shaping code   \\
$k_{\rm s}+r_{\rm s}=n_{\rm s}$             &  number of sign bits     \\
$k_{\rm a}+r_{\rm a}+k_{\rm sh}=n_{\rm s}(m-1)$    & number of amplitude bits \\ \hline
\end{tabular}
\end{table}

The coset shaping is determined by the code generator matrix in shaping-oriented form shown in Fig.~\ref{fig:shmatr} for the case $k_{\rm s}=0$ (no message bits are transmitted as sign bits). 

\begin{figure}[!t]
\begin{center}
\includegraphics[width=90mm]{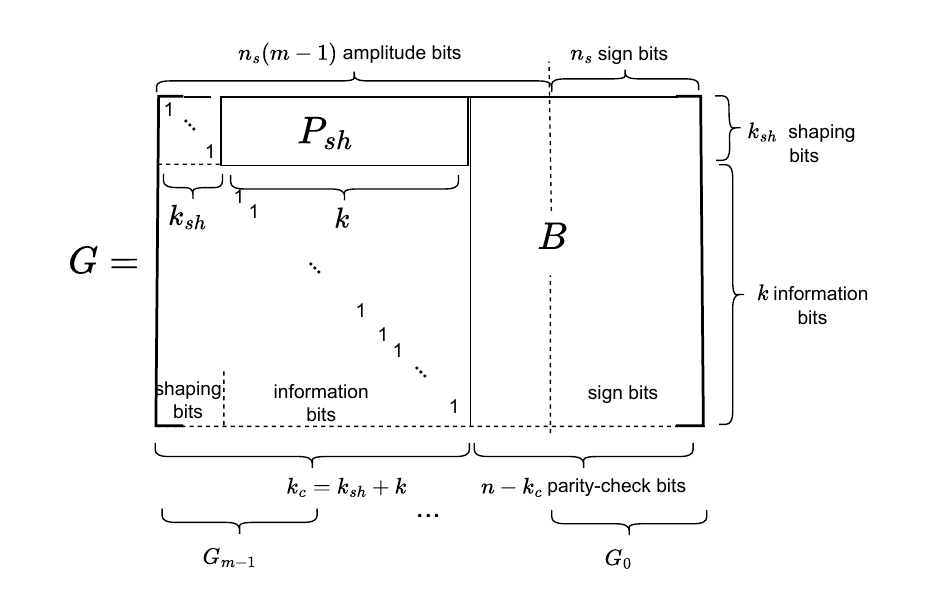}   
\caption{\label{fig:shmatr} Shaping-oriented form of the code generator matrix, $k_{\rm s}=0$.}
\end{center}
\end{figure}

The code design steps are:
\begin{enumerate}
 \item Start with a size $k_{\rm c} \times n$, $k_{\rm c}=k+k_{\rm sh}$ , $n=m n_{\rm s}$ generator matrix $\mat{G}$ given in the systematic form 
$ \begin{pNiceMatrix}
\mat{I}_{k_{\rm c}} ~ \mat{B}
\end{pNiceMatrix}$.
\item Choose a  $k_{\rm sh} \times (n_{\rm sh}=k_{\rm sh}+k)$ generator matrix of a shaping code $\mat{G}_{\rm sh}=\begin{pNiceMatrix}
 \mat{I}_{k_{\rm sh}} ~ \mat{P}_{\rm sh} 
 \end{pNiceMatrix}$ in the systematic form. Let  $k_{\rm s}=0$, $k=k_{\rm a}$.
 \item  
Reduce the first $k_{\rm sh}$ rows of $\mat{G}$ to the form shown in Fig.~\ref{fig:shmatr} by the linear combinations of rows of $\mat{G}$ which add up to $\mat{G}_{\rm sh}$.  This is always possible because of the first block of $\mat{G}$ being $\mat{I}_{k_c}$.  The remaining $n-n_{\rm sh}$ coordinates are also combined respectively. The resulting $k_{\rm sh} \times n$  submatrix $\mat{G}_l$ in the shaping-oriented form of the generator matrix is obtained as $\mat{G}_l \eqdef \mat{G}_{\rm sh}\mat{G} \in \Field_2^{k_{\rm sh} \times n}$. 

For long codes, the use of a generator matrix as in Fig.~\ref{fig:shmatr} is limited due to the high computational complexity of the corresponding shaping scheme.  A form of the generator matrix for a low-complexity version of the encoder in Fig.~\ref{fig:preshaping} is shown in Fig.~\ref{fig:lmatr} for $r_{\rm a}=0$ and $k_{\rm s}>0$. Instead of one $[n_{\rm sh},k_{\rm sh}]$ shaping code, the shaping scheme for a long $[n,k]$-code is based on using $N_{\rm b}=k_{\rm a}/k_{\rm sh}^{\rm b}$, $N_{\rm b} \in \Integers$, short shaping $[n_{\rm sh}^{\rm b},k_{\rm sh}^{\rm b}]$-codes. A parity-check matrix of the short shaping code is denoted as $\mat{G}_{\rm sh}^{\rm b}=\begin{pNiceMatrix} \mat{I}_{\rm sh}^{\rm b} ~ \mat{P}_{\rm sh}^{\rm b} \end{pNiceMatrix}$. This step is performed on blocks of length $n_{\rm sh}^{\rm b}$ bits.   

\begin{figure*}[!t]
\begin{center}
\includegraphics[width=90mm]{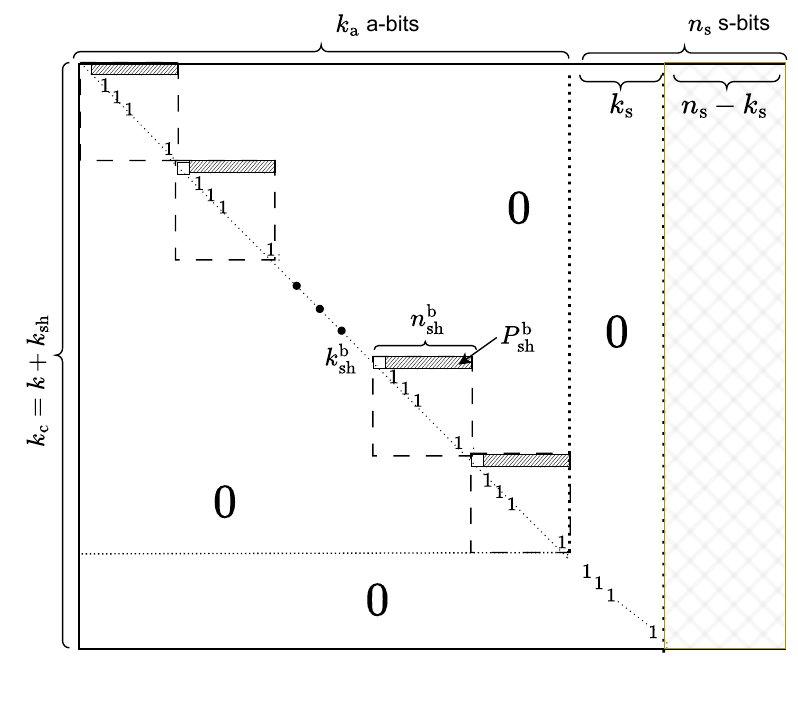}   
\caption{\label{fig:lmatr} Shaping-oriented form of the code generator matrix in example \ref{ex4} for a low-complexity version of the encoder, $r_{\rm a}=0$, $k_{\rm s}>0$.}
\end{center}
\end{figure*}

\end{enumerate}

We achieve a shaping gain at the cost of a slight reduction in the code rate $R_{\rm c}=\nicefrac{k_{\rm c}}{n}$. As before, we use $k$ to denote the message length. The code dimension is denoted by $k_{\rm c}$, and the bit cost of shaping is denoted as $k_{\rm sh}$, therefore, $k = k_{\rm c}-k_{\rm sh}$, where $k_{\rm sh}=N_{\rm b}k_{\rm sh}^{\rm b}$ for long codes. The code rate reduces to $R_{\rm T}=\nicefrac{k}{n}$ bits per binary code symbol or $\tilde{R}_{\rm T}=\nicefrac{k}{n_{\rm s}}$ bits per PAM signal. We will use a tilde to distinguish between rates per binary code symbol and per channel signal. 

Encoding with shaping includes the following steps:
\begin{enumerate}
\item For a given message $\bs u \in \Field_2^k$ and each possible combination of $k_{\rm sh}$ shaping bits $\bs u_{\rm sh}=(u_{{\rm sh},1},\dots , u_{{\rm sh},k_{\rm sh}})$ 
compute a candidate codeword
\begin{equation*}
{\bs v} (\bs u_{\rm sh}, \bs u)  = (\bs u_{\rm sh} ~ \bs u ) \mat{G}.  
\end{equation*}
\item Apply modulation $2^m$-PAM as described in Section \ref{sec:modulation-mapping} to each candidate $\bs u_{\rm sh}$, let 
$\bs s = \psi (\bs v(\bs u_{\rm sh}, \bs u))$ 
be the signal sequence corresponding to one of such candidates. We select the best candidate as one minimizing signal energy
\begin{equation*}
\hat{\bs u}_{\rm sh} =\arg\min_{\bs u_{\rm sh}} \|\psi( \vect{v}(\bs u_{\rm sh}, \bs u)))  \|^2. 
\end{equation*}
Finally, the signal sequence for $\bs u$ is 
\begin{equation*}
\bs s= \psi ( \bs v( \hat{\bs u}_{\rm sh}, \bs u)). 
\end{equation*}

\end{enumerate}
Summarizing, the transmitted signal sequence is obtained by PAM mapping of one of the members of a coset of a linear code determined by $\mat G$. The coset element is determined by $\bs u$, and the leader is determined by  $\bs u_{\rm sh}$. By searching over coset leaders,  the message sequence with the smallest squared Euclidean norm is selected. 

Let $\bs y$ denote the channel output sequence, that is, $\bs y=\bs s+\bs e$, where $\bs e$ is a length $n$ Gaussian vector with independent components. Decoding is performed as follows:
\begin{enumerate}
\item Apply to $\bs y$ a decoder of the binary linear code determined by $\mat G$ 
in the same way as if no shaping is used. 
Denote by $\hat{\bs v}_1^{k_{\rm c}}=({\hat{\bs u}}_{\rm sh}, \hat{\bs v}_{k_{\rm sh}+1}^{k_{\rm c}})$ the 
length $k_{\rm c}=k+k_{\rm sh}$ estimated information part of the codeword.
\item Recover the estimated coset leader 
\begin{equation*}
\bs \xi= \hat{\bs u}_{\rm sh}\mat{G}_{\rm sh}.
\end{equation*}
\item The estimated information part of the codeword is \\ 
${\bs \zeta}=\hat{\bs v}_1^{k_{\rm c}}\oplus \bs \xi.$
The estimated message is 
\begin{equation*}
\hat{\bs u}={{\bs \zeta}}_{k_{\rm sh}+1}^{k_{\rm c}}.
\end{equation*}
\end{enumerate}

For the generator matrix as in Fig.~\ref{fig:lmatr}, the same encoding and decoding steps are performed by $N_{\rm b}$ blocks.    

\begin{example}
\label{ex2}
Consider the generator matrix~\eqref{shorient} of the [6,3]-code in Example \ref{ex1}.
Let $k=2$ and $k_{\rm sh}=1$,  $\mat{G}_{\rm sh}=\begin{pNiceMatrix}1&1&1\end{pNiceMatrix}$. Then 
\begin{equation*}
\mat{G}_l=\mat{G}_{\rm sh} \mat{G}
=\begin{pNiceMatrix}1&1&1&0&0&0  \end{pNiceMatrix},
\end{equation*}
and the shaping-oriented form of the generator matrix is
\vspace{0.2cm}
\begin{equation*}
\mat{G}= \begin{pNiceMatrix}[vlines={4}]
1 & 1 & 1 & 0 & 0 & 0 \\
\cline{1-6} 
0 & 1 & 0 & 1 & 0 & 1 \\
0 & 0 & 1 & 1 & 1 & 0
\end{pNiceMatrix}   . 
\end{equation*}

The matrix $\mat{G}_l$ is used for constructing two cosets of the rate $\nicefrac{2}{6}$ error-correcting code determined by $\mat G$. The corresponding two sets of the 64-QAM signal points are shown in Figs.  \ref{fig:SigSet}A and \ref{fig:SigSet}B. Their average energy is equal to 23 and 19, respectively. 
For each of the four messages, among the two possibilities, a minimum-energy representation from one of two signal sets is selected as shown in Fig.~\ref{fig:SigSet} C. This results in $E=9$. The four minimum-energy signal points that choose sphere shaper \cite{gultekin2020probabilistic} are shown in Fig.~\ref{fig:SigSet} D. This shaper provides $E=9$ too. Notice that the average energy of 8-PAM signals without shaping is equal to 21.

The parameters of this scheme in the notation of Fig.~\ref{fig:preshaping} are $n_{\rm s}=2$, $R_{\rm T}=\nicefrac{1}{3}$, $\tilde{R}_{\rm T}=1$, $R_{\rm c}=\nicefrac{1}{2}$, $R_{\rm sh}=\nicefrac{1}{3}$, $k_{\rm a}=2$, $k_{\rm s}=0$, $r_{\rm a}=1$, $r_{\rm s}=2$. Mapping of a codeword to PAM signals is illustrated in Fig.~\ref{fig:TwoEx}a. \hfill \exampleend
\end{example}  
 
In Example~\ref{ex2}, one of the a-bits is a redundant bit of a codeword. Shaping for parity bits in the framework of DM shaping is discussed in \cite{bocherer2019probabilistic,lentner2025efficient}. In the case of our coset shaping, extending it to parity bits does not require any modification to the scheme.    

\begin{figure}[!t]
\begin{center}
\includegraphics[width=85mm]{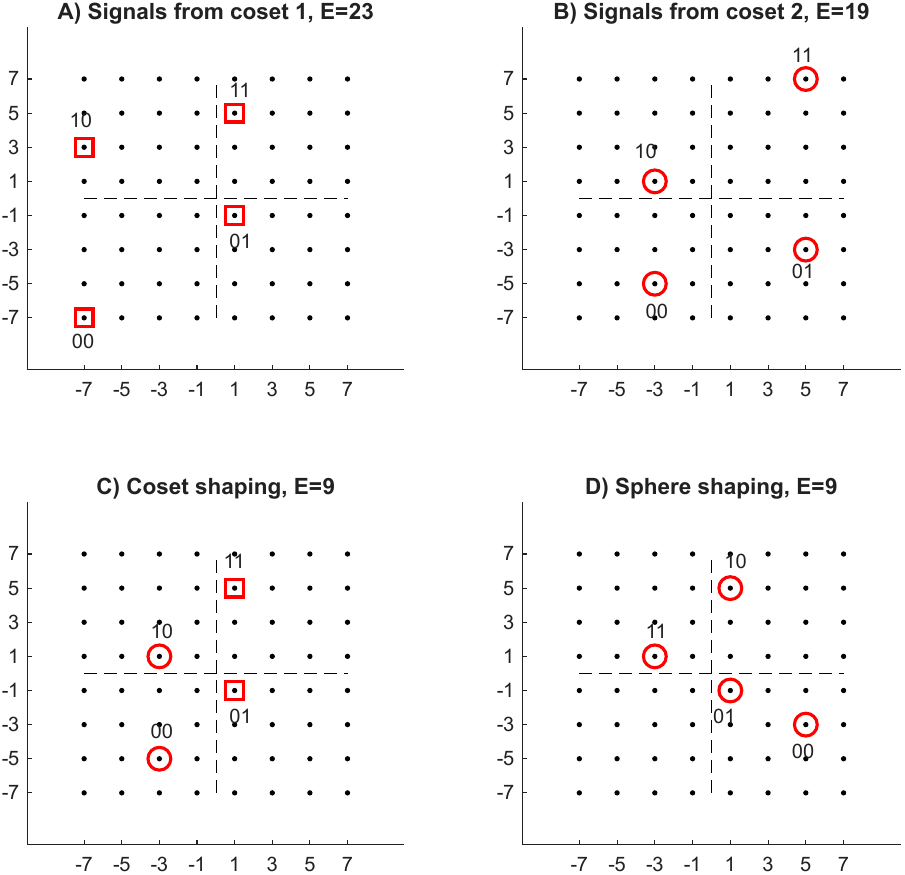}
\caption{\label{fig:SigSet} Two-dimensional shaped Gray-coded modulation. 
The sixty-four black dots show 64-QAM points.
The eight signal points corresponding to the eight codewords are represented by two coset images, one by red squares in Fig.~A, and the other by red circles in Fig.~B. Each signal point is labeled by the message bits. The non-linear sphere shaper in Fig.~D chooses four minimum-energy points among all eight signal points. }
\end{center}
\end{figure}

\begin{figure}[!t]
\begin{center}
\includegraphics[width=85mm]{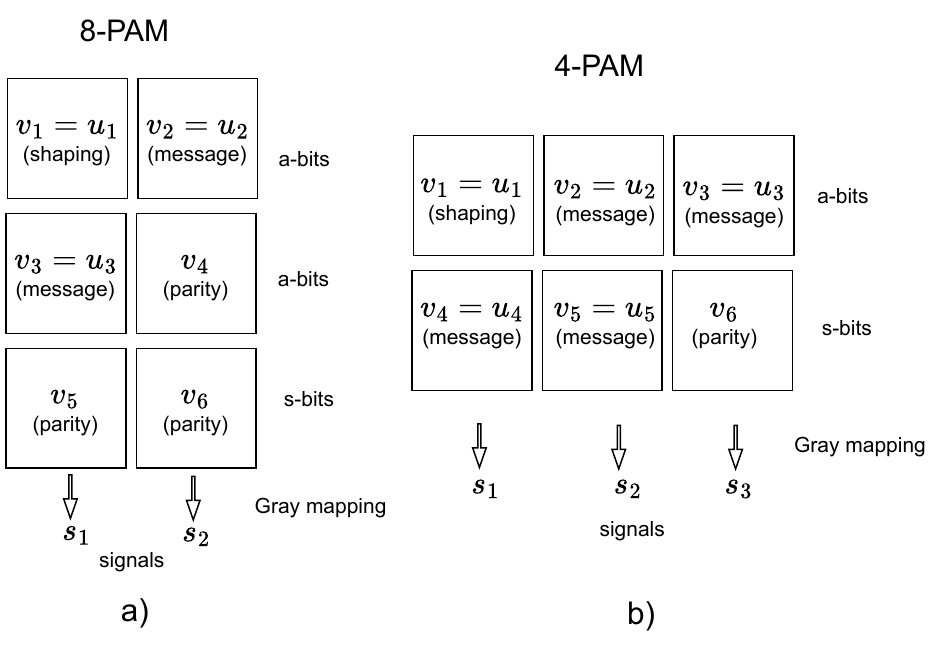}
\caption{\label{fig:TwoEx} Mapping of codewords to PAM signals: a) Example~\ref{ex2}, b) Example~\ref{ex3}.}
\end{center}
\end{figure}

For higher transmission rates, message bits are mapped to sign bits, as demonstrated in the following example.  
\begin{example}
\label{ex3}
Consider 4-PAM transmission, $R_{\rm T}=\nicefrac{2}{3}$, and the generator matrix~\eqref{shorient} of the $[6,5]$-code 
in 
the form 
\begin{equation*}
\mat{G}= \begin{pNiceMatrix}[vlines={6}]
1& 1& 1& 0& 0& 1 \\
\cline{1-6} 
0& 1& 0& 0& 0& 1 \\
0& 0& 1& 0& 0& 1 \\
0& 0& 0& 1& 0& 1 \\
0& 0& 0& 0& 1& 1
\end{pNiceMatrix}.
\end{equation*}
Parameters of the scheme: $n_{\rm s}=3$, $R_{\rm c}=\nicefrac{5}{6}$, $R_{\rm sh}=\nicefrac{1}{3}$, 
$k_{\rm a}=2$, $k_{\rm s}=2$, $k_{\rm sh}=1$, $r_{\rm a}=0$, $r_{\rm s}=1$.
The average energy per signal of the shaped sequences is equal to 3, while the average energy of 4-PAM signals is equal to 5. Mapping of a codeword to PAM signals is shown in Fig.~\ref{fig:TwoEx}b. \hfill \exampleend
\end{example}

The next example illustrates the parameters of the shaping scheme for the low-complexity version of the encoder.

\begin{example}
\label{ex4}
Consider 8-PAM (64-QAM) transmission, $R_{\rm T}=\nicefrac{2}{3}$, $\tilde{R}_{\rm T}=2$. Following shaping oriented form and notation from Fig.~\ref{fig:lmatr}, we choose the $[n,k_{\rm c}]$ ECC with $n=648$ and $k_{\rm c}=486$. $N_{\rm b}=6$ shaping $[n_{\rm sh}, k_{\rm sh}]=[72,9]$-codes are used in order to shape $k_{\rm a}=432$ amplitude bits. Sign bits are formed by $k_{\rm s}=54$ message bits and $r_{\rm s}=n-k_{\rm c}=n_{\rm s}-k_{\rm s}=162$ check bits of the ECC. The first $\nicefrac{k_{\rm a}}{2}=216$ bits are mapped to the first amplitude bits of the 8-PAM signals, and the next $\nicefrac{k_{\rm a}}{2}=216$ bits are mapped into the second amplitude bits. The $n_{\rm s}=216$ bits form signs of the modulation signals.  Mapping of a codeword to 8-PAM signals is shown in Fig.~\ref{fig:TwoEx}a.

The average energy per signal of the shaped sequences in this scheme is equal to
$E=12.15$, while the average energy of 8-PAM signals is equal to 21. \hfill \exampleend
\end{example}

\section{Asymptotic analysis  
\label{analysis}}    

We analyze next a random shaping scheme for coded PAM signaling over the AWGN channel with noise variance~$\sigma^2$. First, we will establish the notation and assumptions. 

\begin{definition}
 An ensemble of coset codes 
 \begin{equation}\label{eq:ensemble-coset-codes}
\set{C}_{\rm c}\eqdef\{\code{C}_{\rm c}=\code{C} \oplus \vect{a}: \code{C} \in \set{C}, \vect{a} \in \Field_2^n \},   
\end{equation}
  is a set of cosets of $[n,k_{\rm c}]$ uniformly random binary linear codes $\code{C}$, where the codes are generated by matrices selected uniformly from the set $\set{C}$ of $k_{\rm c} \times n$ matrices over $\Field_2$, and shift $\vect{a}$ is chosen at random from $\Field_2^n$.     
\end{definition}
The random shift $\bs a$ is introduced to provide uniform distribution over codewords. This property will be later used in Proposition~\ref{pro:2}.

Consider a random shaping scheme with a $[n,k_{\rm c}]$ coset code~$\code{C}_{\rm c}$ and its $[n,k_{\rm sh}]$ shaping subcode $\code{C}_{\rm sh} \subset \code{C}_{\rm c}$, with $k_{\rm sh}=k_{\rm c}-k'$. We call the pair $(\code{C}_{\rm c},\code{C}_{\rm sh})$ {\em shaping construction}. Elements of the construction are $2^{k_{\rm sh}}$ subsets of $\code{C}_{\rm c}$, 
each subset contains $2^{k'}$ codewords from $\code{C}_{\rm c}$. In the previous examples $k'=k$, but in the asymptotic analysis, we assume $k'< k$. The rate loss $\nicefrac{(k-k')}{n}\to 0$ if $n\to \infty$. 

The generator matrix $\mat{G}_{\rm sh}$ of $\code{C}_{\rm sh}$ is a random subset of $k_{\rm sh}$ rows of the generator matrix $\mat{G}$ of $\code{C}$. The remaining part of the matrix, $\tilde{\mat{G}}$, is of size $k'\times n$. In Fig.~\ref{fig:shmatr}, $k'=k$, the top $k_{\rm sh}$ and bottom $k$ rows are the matrices $\mat{G}_{\rm sh}$ and $\tilde{\mat{G}}$, respectively.

Recall $\psi: \Field_2^n \to \set{X}^{n_{\rm s}}$ as in Section~\ref{sec:modulation-mapping}, where $m$ is divisor of $n$. Denote by  $\mathcal R \subseteq \set{X}^{n_{\rm s}} = \set{X}^{\nicefrac{n}{m}}$ a preconstructed {\em shaped set}, $|\mathcal R|=2^k$. In the asymptotic analysis, it is assumed that only signal sequences in  $\mathcal R \subseteq \set{X}^{n_{\rm s}}$ are used for data transmission. In other words, the message $\bs u \in \Field_2^{k'}$ is first  encoded as $\vect{v}=\bs u \tilde{\mat{G}}$, and then $\bs s(\bs u) \eqdef \psi (\bs v \oplus \bs v_{\rm sh} \oplus \bs a)$,
where $\bs a$ is the shift as in~\eqref{eq:ensemble-coset-codes}, $\bs v \oplus \bs v_{\rm sh} \in \code{C}$, and $\bs v_{\rm sh}$ is one of $2^{k_{\rm sh}}$ coset leaders used to transmit the message $\bs u$ over the channel, iff $\bs s(\bs u) \in \mathcal R$. If none of the $2^{k_{\rm sh}}$ signal sequences are in $\mathcal R$, an arbitrary sequence of $\mathcal R$ is transmitted, and we get an error.

We introduce a set of errors 
\begin{equation*}
{\rm E}_0 \eqdef {\rm S}_{n_{\rm s}}(\bs x,\sqrt{n_{\rm s} (\sigma^2+\delta))},
\end{equation*}
where ${\rm S}_{n}(\bs x,\rho)$ denotes the $n$-dimensional sphere of radius $\rho$ around the point $\bs x$, $\sigma$ is the AWGN standard deviation, $\delta>0$. Notice that the center $\bs x$ in ${\rm S}_{n}(\bs x,\rho)$ not necessarily coincides with some signal point in $\set{X}^{n_{\rm s}}$. In the analysis below, instead of searching for the signal sequence $\bs s$ closest to the channel output, we apply the following decoding rule: \emph{the decision $\bs u$ is made if for a unique $\bs u$ the corresponding signal
$\bs s (\bs u) \in \bs y+{\rm E}_0$. Otherwise, an error event is declared.}

The random shaping scheme consists of two random ensembles:
the ensemble of coset codes $\set{C}_{\rm c}$ as in~(\ref{eq:ensemble-coset-codes}) and the 
set $\set{P}$ of random choices of $k_{\rm sh}$ rows among $k_{\rm c}$ rows of the generator matrix $\mat{G}$ of $\code{C}_{\rm c}$. The error probability is estimated as an average error probability on the product of two ensembles ${\set{C}_{\rm c}} \times\set{P}$. An element of ${\set{C}_{\rm c}} \times\set{P}$ is a shaping construction $(\code{C}_{\rm c},\code{C}_{\rm sh})$. It determines 
$2^{k_{\rm sh}}$ subsets each of size~$2^{k'}$. 

We assume that with high probability
 channel output sequences belong to 
\begin{equation}\label{eq:Lambda}
\Lambda=\Lambda(\mathcal R, {\rm E}_0)=\{\bs s+\bs e: \bs s \in \mathcal{R},\bs e\in {\rm E}_0 \}.
\end{equation}

\begin{definition} \label{def:NSM}
We define the normalized second moment (NSM) $G(\Lambda_{\rm a})$ of 
an arbitrary region $\Lambda_{\rm a} \subseteq \mathbb R^{n_{\rm s}}$ as follows
\begin{equation}
G(\Lambda_{\rm a}) = 
\frac{
\frac{1}{n_{\rm s}}\int_{\bs x \in\Lambda_{\rm a}} \|\bs x\|^2 d\bs x
}{\vol{\Lambda_{\rm a}}^{1+\nicefrac{2}{n_{\rm s}}}}.
\end{equation}
\end{definition}

In~\cite{bocharova2025coset}, we have formulated and outlined a draft of the proof for the following theorem, which is now presented in further detail.

\begin{theorem} \label{th}
In the product ensemble $\set{C}_{\rm c} \times \set{P}$, for large enough $m$ ($=\log_2 M$) and $n$ there exists a shaping construction, 
such that an arbitrarily small error probability in AWGN channel with parameters $(P,\sigma^2)$ can be achieved if the code rate per signal dimension $\tilde{R}=\nicefrac{k}{n_{\rm s}}$ satisfies
\begin{equation}\label{eq:theorem}
\tilde{R}\le  \frac{1}{2} \log_2\left( 1+ \frac{P}{\sigma^2} \right)
- \log_2\left(2\pi e ~ G(\Lambda)\right) -o(n), 
\end{equation}
where $G(\Lambda)$ denotes the normalized second moment (NSM) of 
$\Lambda\eqdef\Lambda(\mathcal R, {\rm E}_0)$, and $o(n)\to 0$ when $n\to \infty$.
\end{theorem}

In the following, we give an extended proof of the theorem. 

Although we are considering a linear ECC, the coded modulation signal constellation is nonlinear. First, we will show that in the ensemble  $\set{S}$ of signal points obtained from the ensemble of coset codes ${\set{C}_{\rm c}}$, the probability of any subset is determined only by the number of points in it.

\begin{proposition} \label{pro:2}
The  ensemble of $[n,k_{\rm c}]$-coset codes ${\set{C}_{\rm c}}$  determines the ensemble $\set{S}$ of signal points $\bs s \in \set{X}^{n_{\rm s}}$, where each~$\bs s$ has probability $M^{-n_{\rm s}}$, $M=2^m$.   
\end{proposition}

\begin{proof} 
To verify the proposition statement, we follow the arguments in \cite{gallager1968information}. We first show that in the ensemble ${\set{C}_{\rm c}}$, where symbols 0 and 1 of the  $k_{\rm c} \times n$  generator matrices are selected uniformly at random, the probability of any codeword is $2^{-n}$.

There exist $2^{n(k_{\rm c}+1)}$ choices of the generator matrix $\mat{G}$ and  shift $\bs a$. Each pair has probability $2^{-n(k_{\rm c}+1)}$. Then the probability of a given codeword can be computed as
\[2^{nk_{\rm c}}2^{-n(k_{\rm c}+1)}=2^{-n}.\]

Next, we consider an ensemble of signal points obtained from codewords in ${\set{C}_{\rm c}}$. Given that the map $\psi$ is a bijection on each block of $m$ bits we conclude that each signal point $\bs s$ has probability $\left(2^{m}\right)^{-n/m}$=$M^{-n_{\rm s}}$. 
\end{proof}

\begin{lemma}[Averaging Lemma] \label{lemma1} %
In the ensemble of signal sets
$\set{S}=\{\code{S} \eqdef \psi(\code{C}_{\rm c}): \code{C}_{\rm c} \in \set{C}_{\rm c} \}$, where $\set{C}_{\rm c}$ is an ensemble of $[n,k_{\rm c}]$ coset binary linear codes, for any set 
$\set{W} \subseteq\code{S}$,
the average probability over code images $\code{S}$ is
\begin{equation}
\label{eq:probW}
    \overline{\Pr({\code{S}}\cap \set{\set{W}}| 
     \code{S} )} = 
     \frac{|\set{W}|}{|\code{S}|}.
\end{equation}
\end{lemma}

\begin{proof}
From Proposition \ref{pro:2} it follows that 
for two sets $\set{A}$, $\set{B}$, such that $\set{A} \subseteq \set{B}$, the elements of $\set{A}$, $\set{B}$ are uniformly distributed in the same space. Since $\set{A} \subseteq \set{B}$ then {$\Pr(\set{A} \cap \set{B})=\Pr(\set{A})$}. We obtain
\begin{equation*}
\Pr(\set{A}|\set{B})=\frac{\Pr(\set{A}\cap \set{B})}{\Pr(\set{B})}= \frac{\Pr(\set{A})}{\Pr(\set{B})}=\frac{|\set{A}|}{|\set{B}|}.
\end{equation*}
This yields~\eqref{eq:probW} and completes the proof.
\end{proof}

All sets in Lemma~\ref{lemma1} are subsets of $\set{X}^{n_{\rm s}}$, where elements of the discrete set $\set{X}$ are odd integers from $-2^m+1$ to $2^m-1$.    
We need a similar statement for sets of real-valued channel output sequences $\bs y \in \set{T}$,  
$\set{T} \subseteq \mathbb{R}^{n_{\rm s}}$,  
not necessarily for sets of signal points.  For our purposes, it is sufficient to consider a special case when $\set{T}$ is a sphere.

\begin{remark} \label{sphere}
For a real $\alpha>0$, let ${\rm vol}({\rm S}_{n_{\rm s}}(\bs x,\alpha n_{\rm s}))$ denote the volume of 
the $n_{\rm s}$-dimensional sphere of radius $\alpha n_{\rm s}$ around its center~$\bs x$. 
Denote by 
\begin{equation*}
N_{n_{\rm s}}(\bs x,\alpha n_{\rm s})=\left|{\rm S}_{n_{\rm s}}(\bs x,\alpha n_{\rm s}) \medcap \set{X}^{n_{\rm s}}\right|
\end{equation*}
the number of integer-valued signal points inside this sphere. 
The number of points depends on the location of $\bs x$ with respect to the nearest signal point. The influence of this factor can be reduced by considering high-order modulation, which makes the distance between signal points much smaller than the radius $\alpha n_{\rm s}$ of the sphere.
From \cite[eq.(1.3)]{mazo1990lattice} for large enough modulation order $M=2^m$ follows that for $n_{\rm s}\to \infty$
\begin{equation} \label{eq:s1}
N_{n_{\rm s}}(\bs x,\alpha n_{\rm s})\approx{\rm vol}\left({\rm S}_{n_{\rm s}}(\bs x,\alpha n_{\rm s})\right)2^{-n_{\rm s}}.
\end{equation}
\end{remark}

In order to estimate the average error probability in the product of two ensembles ${\set{C}_{\rm c}} \times\set{P}$, we formulate the error events which occur for a given set of errors ${\rm E}_0$, a random error vector $\bs e$, and a random message $\bs u$. The error is declared if at least one of the three events occurs (see Fig.~\ref{fig:3errors}): 
\begin{IEEEeqnarray*}{rCl}
\mathcal E_\alpha:&  \hspace{5mm}   & \bs e \notin {\rm E}_0 \\
\mathcal E_\beta: & \hspace{5mm} & \mathcal R \medcap \big(\displaystyle\medcup_{\bs v_{\rm sh} \in \code{C}_{\rm sh}} \psi(\bs v \oplus \bs v_{\rm sh} \oplus \bs a) \big)  = \varnothing, \\
\mathcal E_\gamma:&  \hspace{5mm}  &
\bs s(\bs u') \in \bs y + {\rm E}_0, \text{ for some } \bs u'\neq \bs u,
\end{IEEEeqnarray*}
where $\bs y$ is the channel output sequence.

\begin{figure}[!t]
    \centering
    \includegraphics[width=0.8\linewidth]{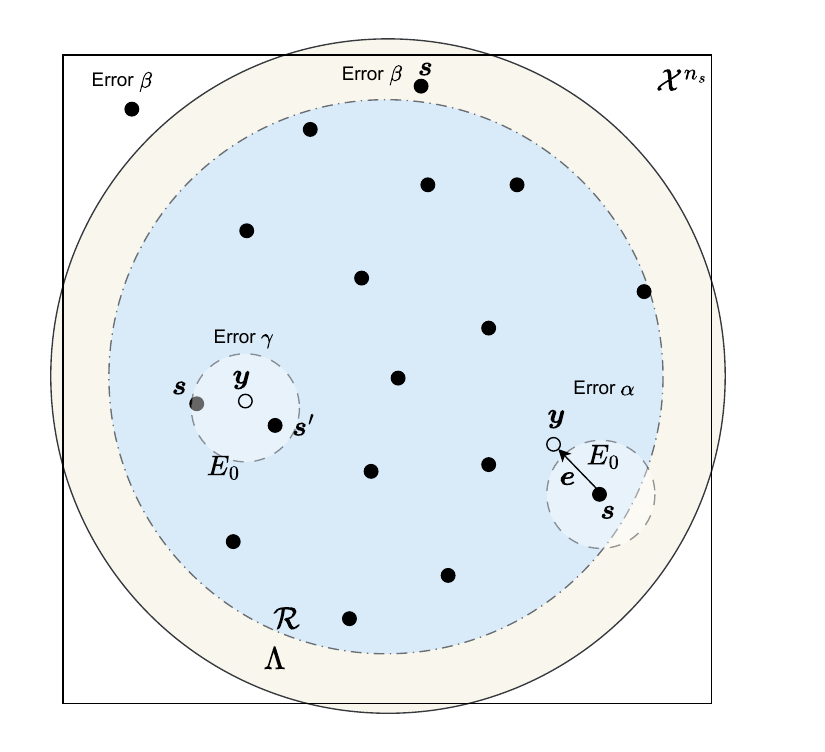}
    \caption{Sets $\set{X}^{n_{\rm s}}$, $\Lambda$, $\mathcal{R}$ and error events $\alpha$, $\beta$, and $\gamma$.}
    \label{fig:3errors}
\end{figure}

In other words, the error is declared in case a non-typical noise vector occurs, or if the encoded signal sequence does not satisfy energy restriction, or if another signal sequence is present in the vicinity of the received sequence $\bs y$. The error event $\mathcal E_\alpha$ depends only on the channel, $\mathcal E_\beta$ depends only on the code, $\mathcal E_\gamma$ depends on both code and channel. See Fig.~\ref{fig:3errors}.


Then, for any $\delta>0$ and any $\epsilon>0$ if $n_{\rm s}$ is large enough, by using Chebyshev's inequality, we obtain  
\begin{align*}
\overline{\Pr(\mathcal E_\alpha) } & ={\Pr(\|\bs e\|^2 \ge n_{\rm s}(\sigma^2+\delta))}\\
& = \Pr(\|\bs e \|^2-n_{\rm s}\sigma^2 \ge 
n_{\rm s}\delta) \le \frac{1}{n_{\rm s}^2\delta^2} <\epsilon,
\end{align*}
where the bar over $\Pr(\cdot)$ denotes the average over the ensemble $\set{C}_{\rm c} \times\set{P}$. 
The probability ${\Pr}(\mathcal E_\alpha)$ does not depend on the choice of the code.

The probability of error events $\mathcal E_\beta$ and $\mathcal E_\gamma$ are random functions of the codes from 
${\set{C}_{\rm c}} \times\set{P}$.

Let $\mathcal{R} \subseteq \set X^{n_{\rm s}}$ of size  $|\mathcal{R}|=2^k$  be a shaping set with average energy
per signal 
$
P=\tfrac{1}{n_{\rm s}}\sum_{\bs s\in \mathcal{R}} \|\bs s\|^2.
$
The probability of the event $\mathcal E_\beta$ follows from Lemma~\ref{lemma2} below.

 \begin{lemma} \label{lemma2}
In the product ensemble ${\set{C}_{\rm c}} \times\set{P}$ of shaping constructions 
for a given shaping region
$\mathcal R \subseteq \set{X}^{n_{\rm s}}$ of size $2^{k}$, there exists a shaping construction  $(\code{C}_{\rm c},\code{C}_{\rm sh})$ such that for any $\bs u {\in \Field_2^k}$, at least one out of the $2^{k_{\rm sh}}$ images $\bs s(\bs u)=\psi(\bs v \oplus \bs a \oplus \bs v_{\rm sh}) $ belongs to  $\mathcal{R}$ with probability arbitrarily close to 1 
{when $n_{\rm s} \to \infty$}.
 \end{lemma}

\begin{proof} 
For any ${\bs v \oplus \bs a} \in \code{C}_{\rm c}$ we compute the probability that for all $\bs v_{\rm sh} \in \code{C}_{\rm sh}$ the image $\psi(\bs v {\oplus} \bs a \oplus \bs v_{\rm sh}) \notin \mathcal R$.  
Let $p$ denote the probability of success, i.e., $\psi(\bs v \oplus \bs a {\oplus} \bs v_{\rm sh}) \in \mathcal R$ for one randomly chosen $\bs v_{\rm sh}$. From~\eqref{eq:probW} in Lemma~\ref{lemma1},
\begin{equation*}
p=2^{-k_{\rm c}}|\mathcal R| = 2^{k-k_{\rm c}}.   
\end{equation*}
Let for $i$-th attempt a random variable $\xi_i$ takes value 1 with probability $p$ and value 0 with probability $1-p$. In the ensemble of random linear codes, codewords are pairwise independent \cite{gallager1968information}.
It means that 
$\Pr(\xi_i,\xi_j)=\Pr(\xi_i)\Pr(\xi_j)$.  Therefore,
\begin{align*}
\overline{\xi_i}=p, \;\; {\rm Var}(\xi_i)=p(1-p),
\;\; \\
{\rm Var}\left(  \frac{1}{2^{k_{\rm sh}}} \sum_{i=1}^{2^{k_{\rm sh}}} \xi_i
\right)=\frac{p(1-p)}{2^{k_{\rm sh}}}.
\end{align*}
Then, for pairwise independent attempts, the probability of the error event $\mathcal E_\beta$ can be estimated as
\begin{align*}
\Pr\left ( \sum_{i=1}^{2^{k_{\rm sh}}} \xi_i= 0 \right) 
&\le
{\Pr}\left (\frac{1}{2^{k_{\rm sh}}} \sum_{i=1}^{2^{k_{\rm sh}}} \xi_i\le 0\right) \\
&\le
\Pr\left (\left|\frac{1}{2^{k_{\rm sh}}} \sum_{i=1}^{2^{k_{\rm sh}}} \xi_i-p\right|\le p\right)\\
&\le
\frac{{\rm Var}(\xi_i)}{2^{k_{\rm sh}} p^2}=\frac{1-p}{2^{k_{\rm sh}}p}<\frac{1}{2^{k_{\rm sh}}p} .
\end{align*} 
Since $2^{k_{\rm sh}}p=2^{k_{\rm sh}}2^{k-k_{\rm c}}=2^{k-k'}$ this probability tends to 0 if $k\to \infty$,  and $k'< k$. 
\end{proof} 

The next step is to estimate the probability of the error event~$\mathcal E_\gamma$. 
The $\mathcal E_\gamma$ happens if some $\bs s=\psi(\bs v)$ is transmitted and the image of $\bs s'=\psi(\bs v')\in \Lambda=\Lambda(\mathcal R, {\rm E}_0)$ as in~\eqref{eq:Lambda}, $\bs v'\neq \bs v$,  
belongs to the noise sphere around $\bs y$. 
Using averaging arguments from Lemma~\ref{lemma1} and Remark~\ref{sphere}, we can write
{
\begin{eqnarray*}
\overline{\Pr(\mathcal E_\gamma,\bs s)} 
&=&\Pr(\bs s'\in {\rm E}_0 \medcap \Lambda \big| \bs s' \in \Lambda)=\frac{{\rm vol}({\rm E}_0)}{\rm vol (\Lambda)},
\end{eqnarray*}
where ${\rm E}_0= {\rm S}_{n_{\rm s}}(\bs y ,\sqrt{n_{\rm s}(\sigma^2+\delta)})
)$ and $\Lambda$ is defined in~(\ref{eq:Lambda}).}


Applying the union bound, we obtain 
\begin{equation}\label {eq:e_Gamma}
\overline{\Pr(\mathcal E_\gamma)} \le 
2^k\frac{{\rm vol}({\rm E}_0)}{\rm vol (\Lambda)}.
\end{equation}

\begin{lemma} 
The average energy of the channel output sequences from the shaped region $\mathcal R$ per signal is 
\begin{IEEEeqnarray}{c} \label{eq:P}
P(\Lambda)\eqdef
 \frac{1}{n_{\rm s}}{\rm Var} (\Lambda)= P+\sigma^2.
\end{IEEEeqnarray}
\end{lemma}
\begin{proof}
Since in the AWGN channel, the noise sequence $\bs e$ does not depend on the transmitted signal sequence $\bs s$, we have that
\begin{equation*}
{\rm Var} (\bs s +\bs e) = {\rm Var} (\bs s) +{\rm Var} (\bs e)= n_{\rm s}(P+\sigma^2). \qedhere
\end{equation*}
\end{proof}

We have now collected all the necessary facts to formulate and prove the main statement of the paper, Theorem~\ref{th}.

\begin{proof}

According to Definition~\ref{def:NSM},
the normalized second moment 
of $\Lambda$ is 
\begin{equation}\label{eq:nsm}
 G(\Lambda) = 
\frac{
\frac{1}{n_{\rm s}}\int_\Lambda \|\bs x\|^2 d\bs x
}{\vol{\Lambda}^{1+\nicefrac{2}{n_{\rm s}}}}
=\frac{P(\Lambda)}{\vol{\Lambda}^{\nicefrac{2}{n_{\rm s}}}}\;, 
\end{equation}
where for $P(\Lambda)$ is determined by (\ref{eq:P}).
Hence,
\begin{IEEEeqnarray}{c}\label{eq:volumeL}
{\rm vol}(\Lambda)=\left(\frac{P(\Lambda)}{G(\Lambda)}\right)^{\nicefrac{n_{\rm s}}{2}}=
\left(
\frac{P+\sigma^2}{G(\Lambda)} 
\right)^{\nicefrac{n_{\rm s}}{2}}.   
\end{IEEEeqnarray} 

For the volume of the even-dimensional sphere of radius $\sigma \sqrt{n_{\rm s}}$, we can use the formula
\begin{IEEEeqnarray}{c} \label{eq:volumeS}
\vol{{\rm E}_0}=\frac{1}{\sqrt{2\pi n_{\rm s}}} (2\pi e \sigma^2)^{\nicefrac{n_{\rm s}}{2}}.
\end{IEEEeqnarray}
Substituting (\ref{eq:volumeL}) and (\ref{eq:volumeS}) into (\ref{eq:e_Gamma}) together with elementary transformations proves the theorem.
\end{proof}

\subsection{Discussion of the asymptotic performance}


From Theorem~\ref{th} and~\eqref{eq:theorem} follows that the achievable rate is below the capacity by a value decreasing with decreasing the NSM of the region $\Lambda$. 
The shape of $\Lambda$, in turn, depends on the choice of the shaping region $\mathcal R$. In the proof of the theorem, we imposed two restrictions on $\mathcal R$: the cardinality, $|\mathcal R|=2^k$, and the average signal energy, $P$. The minimum value of $G(\Lambda)=\nicefrac{1}{2\pi e}$ is achieved if  $\Lambda$ is a sphere, which would be true if the region $\mathcal {R}$ were a sphere too. 

If $\mathcal {R}$ was chosen as the Voronoi region of the origin in Construction A lattice over a large enough finite field, then the achievability of $G(\Lambda)\to\nicefrac{1}{2\pi e}$ would follow, e.g., from \cite[Thm.~3]{erez2005lattices}. Finite-length lattices approaching this limit were considered in  
\cite{zhou2017shaping,yurkov2007random, KudryashovYurkov08}. 

We cannot prove that in our ensemble of PAM-modulated coset shaping codes, the shaping region approaches a sphere with $n$ and $m\to \infty$. However, our experiments with 
codes of different lengths have shown that the shaping gain tends to its ultimate limit as the length of the shaping code increases. These observations led us to the following conjecture. 

\begin{conjecture}
    If $n\to \infty$ then ${G(\Lambda)} \to \nicefrac{1}{2\pi e}$ and $R \to C$, where $C$ denotes the AWGN capacity. 
\end{conjecture}

The scheme in the proof of the main theorem differs from that in the description of the system in Section~\ref{sec:coset} and in simulations. First, instead of a linear code, we considered a coset code (a random shift of a linear code). Second, in the proof, we assumed that the modulation order is large. The reasonable question is: how important are these assumptions for achieving near-optimal system performance? 

Indeed, both assumptions are technical and used to simplify the arguments. Coset codes are used to guarantee that the all-zero codeword and its corresponding signal sequence have the same probability as any other codeword. Due to the high modulation order assumption, we can state that any sphere contains the same number of points of the signal space, regardless of the position of its center.

\section{Bounds on the maximum-likelihood decoding error probability} 
\label{bounds}

We present next a modified  Gallager's random coding bound (RCB) on the error probability of shaped coded modulation signals. First, we revisit  Gallager's random coding bound.   


\subsection{Gallager's random coding bound}

For a random ensemble of the length $n$ linear codes with codeword symbols chosen  independently according to the probability  distribution $Q(x)$, the ensemble average  maximum-likelihood (ML) decoding error probability $P_{\rm e}$ for transmission over a discrete-time memoryless channel is bounded as \cite{gallager1968information}
\begin{equation} \label{eq:RCB}
P_{\rm e}\le \exp \left  \{-nE_{\rm G}(E,\sigma,Q,{\tilde R}_{\rm T}) \right\},
\end{equation}
where 
\begin{align} \label{eq:exp}
& E_{\rm G}(E,\sigma,Q, {\tilde R}_{\rm T})= \nonumber \\
& {\small \max_{ 0\le \rho \le 1} -\ln \left(\int_{-\infty}^\infty \left(\sum_{x \in {\mathcal X}}f_{{\mathcal Y}|{\mathcal X}}(y|x)^{\frac{1}{1+\rho}}   Q(x)\right)^{1+\rho}dy\right) -\rho {\tilde R}_{\rm T}}. 
\end{align}
Observe that $\mathcal X$ denotes  the input alphabet, ${\mathcal Y}$ is channel output alphabet, 
{$E=\sum_{x\in \mathcal X} x_i^2 Q(x_i)$} 
is signal energy, $\sigma$ is noise standard deviation, ${\tilde  R}_{\rm T}$ is transmission rate, $f_{{\mathcal Y}|{\mathcal X}}(y|x)$ is pdf of the channel noise and now, $E_G$ denotes Gallager's exponent.
Assume that we transmit information by $M$-PAM signaling over the additive white Gaussian noise (AWGN) channel, that is 
$y_i=x_i+n_i$, where $E={\mathbb{E}}\{x_i^2\}$, $\mathbb{E}\{\cdot\}$
denotes mathematical expectation,
${\mathcal X}=\{-M-1,...,-3,-1,1,3,...M-1\}$, and the noise variance is $\sigma^2$. Then the signal-to-noise ratio is $\textnormal{SNR}=\nicefrac{E}{\sigma^2}$ and 
\begin{equation*}
f_{{\mathcal Y}|{\mathcal X}}(y|x)=f_N(y-x)=\frac{1}{\sqrt{2\pi}\sigma}\exp\{-\nicefrac{(y-x)^2}{2\sigma^2}\}.
\end{equation*}

Next, we compare three approximations of the RCB: 
\begin{itemize}
 \item[a)]the
bound is computed for the probability distribution for which the channel capacity is achieved (RCBC);
 \item[b)] the probability
distribution is optimized for the error probability $P_{\rm e}$ = $10^{-6}$
(RCBE);
\item[c)] the error exponent is optimized for each SNR value
(RCBO).
\end{itemize}

In Figs. \ref{fig:gal_short16}-\ref{fig:gal_long}, bounds on the FER of ML decoding for shaped and unshaped coded PAM (QAM) signaling for these three scenarios are shown. For comparison, the simulation results are given on the same figures. In all five figures, {in captions},
we also present the ultimate SNR limit (${\rm SNR}_{\rm limit}$) for the corresponding transmission rate and modulation. 

\begin{figure}[!t]
\centering
\includegraphics[width=85mm]{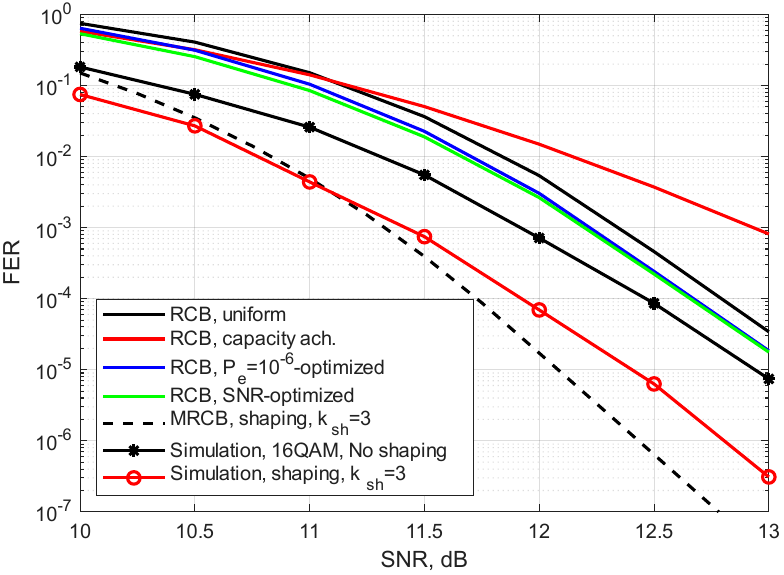}   
\caption{Gallager's random coding bound on the FER of ML decoding for the length 128 bit codes with unshaped and shaped 16QAM signaling,  ${\tilde R}_{\rm T}=3.0$ bpcu, and the ordered statistic decoding (OSD) simulation results, $\textnormal{SNR}_\textnormal{limit}=8.86$ dB. }
\label{fig:gal_short16} 
\end{figure}

\begin{figure}[!t]
\centering
\includegraphics[width=85mm]{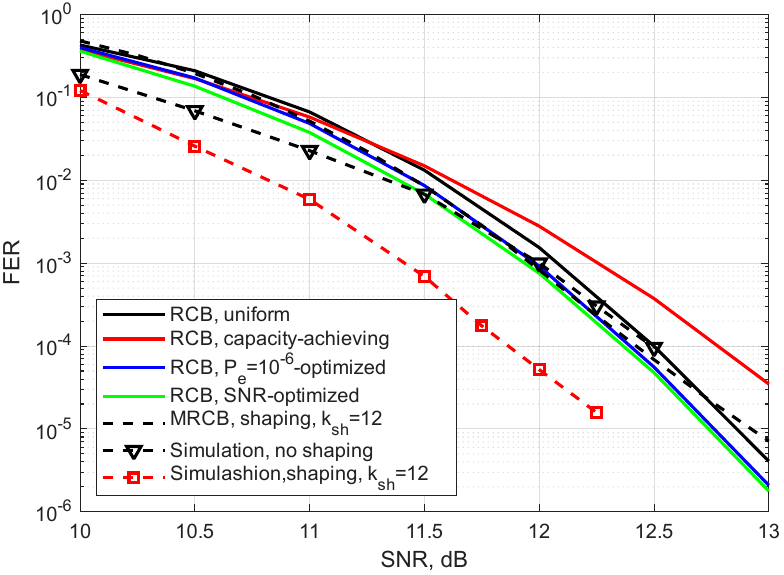}   
\caption{Gallager's random coding bound on the FER of ML decoding for the length 192 bit codes with unshaped and shaped 64QAM signaling, ${\tilde R}_{\rm T}=3.0$ bpcu, and the ordered statistic decoding (OSD) simulation results, $\textnormal{SNR}_\textnormal{limit}=8.45$ dB. }
\label{fig:gal_short64} 
\end{figure}

\begin{figure}[!t]
\centering
\includegraphics[width=85mm]{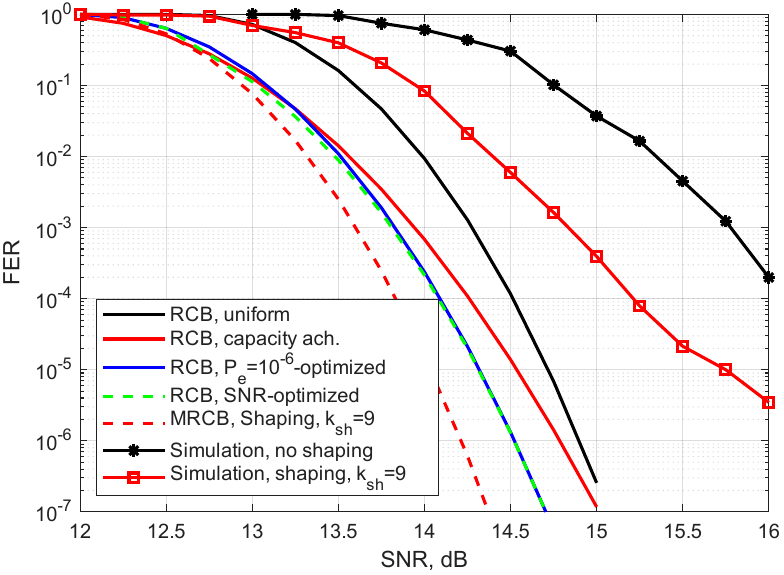}   
\caption{Gallager's random coding bound on the FER of ML decoding for the length 648 bit codes with unshaped and shaped 64QAM signaling, ${\tilde R}_{\rm T}=4$ bpcu, and simulation results, $\textnormal{SNR}_\textnormal{limit}=11.84$ dB. }
\label{fig:gal_mid648} 
\end{figure}

\begin{figure}[!t]
\centering
\includegraphics[width=85mm]{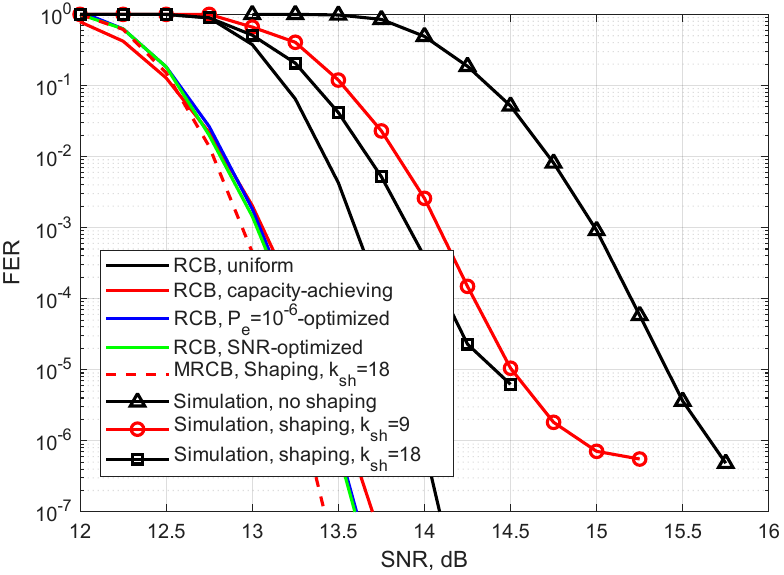}   
\caption{Gallager's random coding bound on the FER of ML decoding for the length 1944 bit codes with unshaped and shaped 64QAM signaling, ${\tilde R}_{\rm T}=4$ bpcu, and simulation results, $\textnormal{SNR}_\textnormal{limit}=11.84$ dB. }
\label{fig:gal_mid1944} 
\end{figure}

For short codes (schemes $S1,S2$ for unshaped signaling and $S3, S4$ for shaped signaling in Table \ref{tab:short}) in Figs. \ref{fig:gal_short16}, \ref{fig:gal_short64}, even the RCBO version of the bound appears to be loose, noticeably larger than the simulation results. While the RCBO bound almost coincides with RCBE, the capacity-optimum bound (RCBC) is the worst. Moreover, it is worse than the RCB for the uniform distribution. This observation is unexpected: the capacity-achieving distribution is often considered optimal across all SNRs. According to the RCB, the expected shaping gain is only about 0.15 dB. 
In this scenario, practical codes show better performance than predicted by the RCB, with a practical shaping gain of about 0.5 dB.

For mid-length WiFi codes (schemes $M1$,$M2$ for unshaped signaling and $M3$,$M5$ for shaped signaling in Table \ref{tab:midcod}) in Figs. \ref{fig:gal_mid648}, \ref{fig:gal_mid1944}: the differences among the three versions are minor, though the RCBC remains the weakest. However, as expected, these three bounds are better than the RCB for the uniform distribution. Practical low-complexity decodable quasi-cyclic (QC)-LDPC codes lose about 1 dB with respect to the RCB. The expected shaping gain according to the RCB is of order 0.5 dB.  

For $n=64800$-long codes,  Fig.~\ref{fig:gal_long}, the three versions of the RCB 
become almost identical. The expected shaping gain exceeds 1 dB. At $P_{\rm e}=10^{-6}$, optimal coding can be only 0.3 dB from the Shannon limit. Our simulation results are 1.0 dB away from the RCB. 

\subsection{Modified Gallager's random coding bound}
Coset shaping is applied to the amplitude bits, reducing the energy of the transmitted signal. The energy gain 
\begin{equation*}
g=\frac{n_{\rm sh}E_{\rm PAM} }{{(m-1)\rm E}\{\|\psi( \vect{v}(\bs u_{\rm sh}, \bs u))) \|^2\}}, 
\end{equation*}
where $E_{\rm PAM}=\tfrac{M^2-1}{3}$ and the mathematical expectation is numerically evaluated during the encoding simulation. Averaging is performed over all codewords. In the case of mid-length and long LDPC codes, the gain is estimated for blocks of length equal to the length of the shaping code. In Tables \ref{tab:short}-\ref{tab:NBLDPC}, we give values of $g$ for short, mid-length, and long codes used for simulations. 

In order to compute the modified RCB bound for shaped coded modulation (MRCB), we rewrite (\ref{eq:RCB}) as 
\begin{equation} \label{eq:RCBSCM}
P_{\rm e}\le \exp \left  \{-nE_{\rm G}(E_{\rm av},\sigma,
Q_{\rm u}, {\tilde R}_{\rm T}) \right\},
\end{equation}
where for $E_{\rm G}(\cdot)$ we use formula (\ref{eq:exp}), $Q_{\rm u}$ denotes uniform distribution over the set of signal points, and $E_{\rm av}=\nicefrac{E_{\rm PAM}}{g}$. The MRCB for the three scenarios is shown in  {Figs.~\ref{fig:gal_short16}--\ref{fig:gal_long}}. This bound has less gap to the simulation results than  all versions of Gallager's bound.

\begin{figure}[!t]
\centering
\includegraphics[width=85mm]
{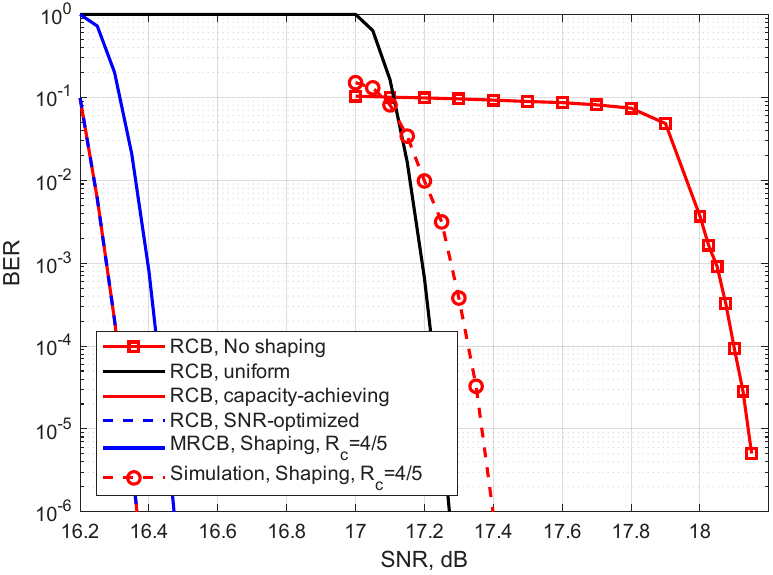}
\caption{Gallager's random coding bound on the FER of ML decoding for the length 64800 bits code with unshaped and shaped 64QAM signaling, ${\tilde R}_{\rm T}=5.33$ bpcu, and simulation results, $\textnormal{SNR}_\textnormal{limit}=15.96$ dB. }
\label{fig:gal_long} 
\end{figure}

\section{Simulation results \label{sec:numres}}

All comparisons are given in terms of QAM-modulation, as shaped QAM performance is studied in most technical papers.

Considering that our proposed coset shaping is a geometric shaping, it is natural to compare our results with Voronoi shaping, where we choose the state-of-the-art work in~\cite{li2025coded}. We will refer to the following Voronoi constellations. 
\begin{description}
    \item[${\rm E}_8^{24}$] Denotes a $8$-dimensional constellation, based on the partition $\Integers^8/8{\rm E}_8$. This means that the shaping lattice is $8{\rm E}_8$, seeing as a sublattice of $\Integers^8$. The partition order $M=2^{24}$ implies that it carries $24$ bits per $8$-dimensional symbol. 
     \item[$\Lambda_{24}^{72}$] Denotes a $24$-dimensional constellation, based on the partition $\Integers^{24}/2\Lambda_{24}R_{24}$, where the shaping lattice $\Lambda_{24}R_{24}$ is a scaled version of the Leech lattice $\Lambda_{24}$. The partition order $M=2^{72}$ implies that it encodes $72$ bits per $24$-dimensional symbol. 
\end{description}
Furthermore, the mention of \emph{hybrid} means an alternative for the Gray mapping, which combines set-partition~\cite{Ungerboeck82} and pseudo-Gray mapping used by authors~\cite{li2025coded} to increase SNR gains over QAM, for high SNR regimes.

Shaping for short codes is performed according to the scheme shown in Fig.~\ref{fig:shmatr}. Parameters of ECCs and shaping codes are given in Table \ref{tab:short}, where eBCH and sBCH denote extended and shortened BCH codes, respectively, and GCC denotes the \emph{generalized concatenated code}. The fourth-order OSD simulation results are plotted in Fig.~\ref{fig:sim_short}. For comparison, in the same figure, we show the simulation results for 192-bit-long polar codes in \cite{cocskun2019efficient} used with constant composition distribution matching (CCDM).

Shaping for length $n=648$ bit and $1944$ bit codes is performed according to the scheme presented in Fig.~\ref{fig:lmatr}. Parameters of ECCs and shaping codes are given in Table \ref{tab:midcod}. We consider QC-LDPC codes in the WiFi standard, specified by the rate $R_{\rm c}$, the base code of length $n_{\rm b}$, and the lifting factor $L$. The QC shaping codes in the schemes $M3$, $M4$, $M5$, and the schemes $L3$ and $L4$ below
are chosen from tables \cite{Chen:codetables}.
Since the codes for the same parameters are not unique in \cite{Chen:codetables}, for convenience, we present the lists of code generators in Table~\ref{tab: QCcodes}.

The QC codes are determined by their parity-check matrices in the form: 
$\begin{pmatrix}
  \mat{C}_1&\mat{C}_2&...&\mat{C}_{\nicefrac{n}{p}}   
\end{pmatrix}
$,
where $\mat{C}_i$, $i \in [1: \nicefrac{n}{p}]$ is a circulant matrix of size $p \times p$ obtained by $p$ shifts of the generator {$\vect{g}_i$}, and $n/p$ is an integer.  
The FER performance of BP decoding for codes of lengths 648 and 1944 bits, with shaped and unshaped 64-QAM signaling, is shown in Fig.~\ref{fig:gal_mid648} and Fig.~\ref{fig:gal_mid1944}, respectively. 

\begin{figure}[!t]
\centering
\includegraphics[width=85mm]{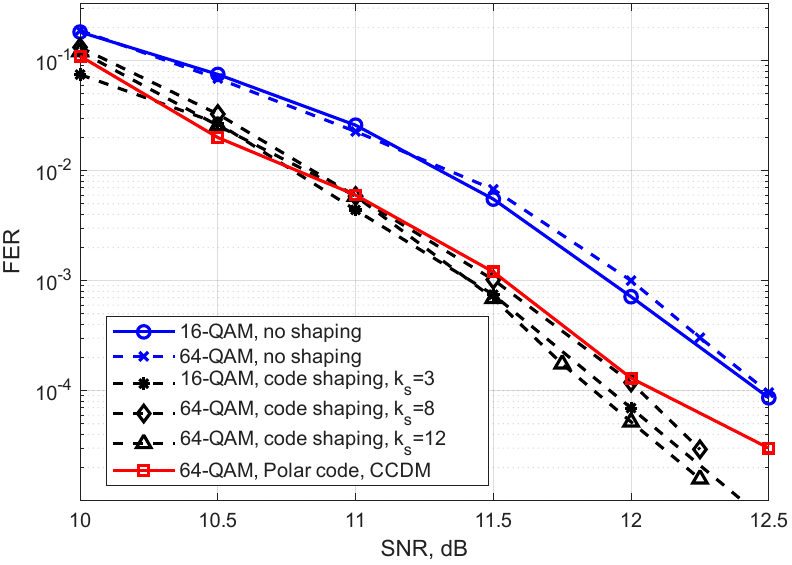}   
\caption{OSD FER performance for the length 128 bit and 192 bit codes with unshaped and shaped 16QAM and 64QAM signaling, respectively, ${\tilde R}_{\rm T}=3.0$ bpcu, in comparison with the shaped length 192 polar codes in \cite{cocskun2019efficient}, $\textnormal{SNR}_\textnormal{limit}=8.86$ dB. }
\label{fig:sim_short} 
\end{figure}

Parameters of length $n=64800$-bit NB QC-LDPC codes used with unshaped and shaped 256-QAM are presented in Table \ref{tab:NBLDPC}. They are determined by the $r_{\rm b}\times n_{\rm b} $ base parity-check matrix, the lifting factor $L$, the size of the alphabet $q$, and the column weight distribution $\lambda(x)$. The NB QC-LDPC codes were constructed as in \cite{bocharova2022design}. Shaping codes were constructed according to the scheme in Fig.~\ref{fig:lmatr}.  

In Fig.~\ref{smul_long}, we compare the BP decoding BER performance for the $n=64800$-bit NB QC-LDPC codes used with unshaped and shaped 256-QAM with the same performance of both binary LDPC codes from the DVB-S2 standard used with 64-QAM shaped as in \cite{li2025coded} and binary LDPC codes from the ATSC standard used with unshaped 64-QAM. The parameters of the simulated communication scenarios are tabulated in Table~\ref{tab:NBLDPC_scen}. Fifty iterations of BP decoding were simulated. In the same figure, we present the ultimate SNR limit for 256-QAM signaling ($\textnormal{SNR}_\textnormal{limit}$) and the SNR limit for unshaped 256-QAM signaling ($\textnormal{SNR}_\textnormal{unif}$) used at ${\tilde R}_{\rm T}$=5.33 bpcu.  
 
\begin{table}[!t]
\center
\renewcommand{\arraystretch}{1.5}
\caption{Parameters of communication scenarios for long codes\label{tab:NBLDPC_scen} }
\begin{tabular}{|c|l|l|c|}
\hline
\textbf{QAM} & \boldsymbol{$R_c$}  & \textbf{Code info}    & S\textbf{haping code}  \\                                   \hline

 256  & \nicefrac{2}{3}  &  ATSC \cite{ATSCstandard}& No shaping    \\ \hline
 256  & \nicefrac{2}{3}  &  NB QCLDPC, $q=2^4$   & No shaping     \\ \hline    
  64    &  \nicefrac{2}{3} & DVB-S2, \cite{standard2014digital}   & hybrid  $E_8^{24}$,\cite{li2025coded}    \\ \hline
 64    &  \nicefrac{2}{3} & DVB-S2, \cite{standard2014digital}   & hybrid  $\Lambda_{24}^{72}$ \cite{li2025coded}   \\ \hline
 256 & \nicefrac{3}{4}  & NB QCLDPC,  $q=2^3$  & coset [135,15,54] \\
\hline         
256 & \nicefrac{4}{5}  &  NB QCLDPC, $q=2^3$   &coset [90,16,32]  \\\hline
\end{tabular} 
\end{table}

\begin{figure}[!t]
\centering
\renewcommand{\arraystretch}{1.5}
\includegraphics[width=85mm]
{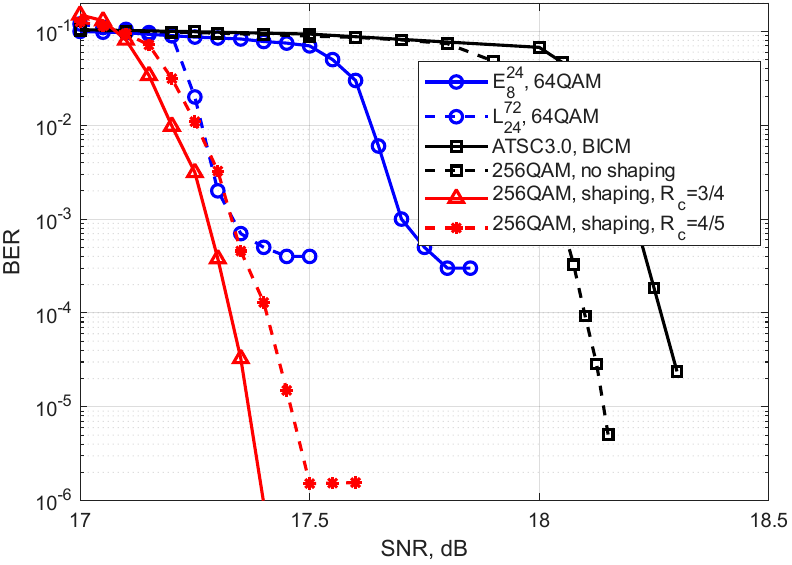}
\caption{Comparison of NB LDPC coded QAM-256 signaling with and  without shaping with  \cite{li2025coded}, ${\tilde R}_{\rm T}=5.33$ bits per QAM signal,  $\textnormal{SNR}_\textnormal{unif}$ is 16.98 dB, $\textnormal{SNR}_\textnormal{limit}$ is 15.96 dB }
\label{smul_long} 
\end{figure}

\section{Conclusion}
A new shaping technique applicable to coded QAM signaling was proposed and analyzed. The introduced shaper belongs to the class of multidimensional geometric shapers and allows shaping both the message and the parity bits of the codeword. Asymptotic analysis of the new technique indicated that it is capacity-achieving as both code length and modulation order tend to infinity. The presented simulation results and comparisons for a variety of error-correcting codes with different rates and lengths, used with shaped PAM signaling, suggest that the proposed technique outperforms known solutions, especially in the error-floor region. It is shown that the modified Gallager bound has a smaller gap with the simulation results than the Gallager bound. Overall, our proposed coset-shaping technique provides strong theoretical guarantees and substantial practical performance improvements for coded QAM signaling systems.






\begin{table}[!t]
\center
\caption{Short-length schemes, ${\tilde R}_{\rm T}=3$ bits/QAM signal, $n_s=64$, OSD-4,
shaping gain is measured at $P_e=10^{-4}$ \label{tab:short} }
\renewcommand{\arraystretch}{1.5}
\begin{tabular}{|c|c|m{1.5cm}|m{2cm}|m{1cm}|}
\hline
Scheme &QAM & $[n,k_c,d]$-code & Shaping code $[k_c,k_{\rm sh},d_{\rm sh}]$, energy gain (dB)
& Shaping gain (dB) \\ \hline
\multicolumn{5}{|c|}{No shaping }   \\ \hline
S1        &  16   &  [128,96,10], eBCH & ---            &   ---      \\ \hline
S2        &  64   & [192,96,28], GCC in \cite{blokh1974coding} & ---           &    ---    \\ \hline        
\multicolumn{5}{|c|}{Coset shaping }   \\ \hline
S3        &  16   &  [128,99,10], eBCH  & [99,3,56], ([7,3,4]$\times$14),
$g=0.783$&    0.62   \\ \hline 
S4        &  64   & [192,108,24], sBCH\cite{helgert1973shortened} 
& [108,12,47] in  \cite{Grassl:codetables}, $g=2.053$    &    0.75     \\ \hline                  
\end{tabular} 
\end{table}

\begin{table}[!t]
\center \renewcommand{\arraystretch}{1.5}
\caption{Schemes based on QC-LDPC codes in the WiFi (IEEE 802.11) standard with  64QAM, ${\tilde R}_{\rm T}=4$ bits/QAM signal, $n_{\rm s}=216$ and 648, 50 decoding iterations, shaping gain is measured at $P_e=10^{-4}$ \label{tab:midcod}}
\renewcommand{\arraystretch}{1.5}
\begin{tabular}{|c|c|c|m{2cm}|c|m{1cm}|}
\hline
Scheme  & $R_c$ & $n_{\rm b}\times L$ &  Shaping code $[n^{b}_{\rm sh},k^{b}_{\rm sh},d_{\rm sh}]$, energy gain (dB)& $N_{\rm b}$ & Shaping gain  (dB) \\ \hline
\multicolumn{6}{|c|}{No shaping }   \\ \hline
M1        &    2/3  &  $24\times 27$   &   ---      &---      &   ---      \\ \hline
M2        &    2/3  &  $24 \times 81$   &   ---     &---      &   ---      \\ \hline        
\multicolumn{6}{|c|}{Coset shaping }   \\ \hline
M3        &  3/4   & $24 \times 27$  & [72,9,32], QC1, $g=2.378$  & 6 &    0.95  \\ \hline 
M4        &  3/4   & $24 \times 81$  & [72,9,32], QC1, $g=2.378$&18  
&    0.88  \\ \hline 
M5        &  3/4   & $24\times 81$  & [144,18,52], QC2, $g=2.425$ & 9  &    1.07    \\ 

 \hline      
 \end{tabular} 
\end{table}
  
\begin{table}[!t]
\center\renewcommand{\arraystretch}{1.5}
\caption{Schemes with long NB QC-LDPC codes, $n\approx 64800$ bits,\\
$R=5.33$ bits/256-QAM signal, 50 decoding iterations, shaping gain is measured at BER=$10^{-4}$ \label{tab:NBLDPC} }
\begin{tabular}{|c|m{0.4cm}|m{2cm}|m{1.7cm}|m{0.3cm}|m{0.8cm}|}
\hline
Scheme  & $R_c$ & ECC: $r_{\rm b} \times n_{\rm b}$, $L,q,\lambda(x)$  & Shaping code $[n^{b}_{\rm sh},k^{b}_{\rm sh},d_{\rm sh}]$, energy gain (dB) & $ N_{\rm b}$ & Shaping  gain (dB) \\ \hline
\multicolumn{6}{|c|}{No shaping }   \\ \hline
L1        &  2/3  &  ATSC \cite{ATSCstandard}& ---   &   ---  & ---    \\ \hline
L2        &  2/3  &  $40 \times 120$,135,$2^4$, $80x^2$+$33x^3$+$7x^{16}$   & ---   & ---  &  --- \\    \hline 
\multicolumn{6}{|c|}{Coset shaping }   \\ \hline
L3       &  3/4  & $30 \times 120$,180,$2^3$, $70x^2$+$49x^3$+$x^{15}$   &    [135,15,54]  QC3, $g=3.011$ &320   & 0.86   \\ \hline
L4       &  4/5  &  $30 \times 150$,144,$2^3$, $90x^2+60x^3$    &[90,16,32], subcode of QC4,
$g=4.254$ 
&480   & 0.85  \\   \hline
\end{tabular} 
\end{table}

\begin{table}[!t]
\center\renewcommand{\arraystretch}{1.5}
\caption{Generators of quasicyclic codes used in simulations \label{tab: QCcodes} }
\begin{tabular}{|c|c|m{0.7cm}|m{3.5cm}|}
\hline
Code & $[n,k,d]$&Circu\-lant size $p$ & Generators in octal form \\ \hline
QC1 &  [72,9,32]&12 &1233, 1247, 1127, 723, 245,55\\ \hline
QC2& [144,18,52]&18 &1, 75523, 147153, 111127, 146657, 16333 51573 45425\\ \hline
QC3& [135,15,54]&15 & 5052, 56364, 76254, 4575, 5536, 4734, 665, 5333, 76744 \\\hline
QC4& [90,18,32]& 18 & 401, 123225, 123633, 122525, 71125 \\\hline
\end{tabular} 
\end{table}


 
%

\bibliographystyle{IEEEtran}
\bibliography{biblio}


 




\vfill

\end{document}